\documentclass[11pt]{article}\usepackage[]{graphicx}\usepackage[]{color}
\makeatletter
\def\maxwidth{ %
  \ifdim\Gin@nat@width>\linewidth
    \linewidth
  \else
    \Gin@nat@width
  \fi
}
\makeatother

\definecolor{fgcolor}{rgb}{0.345, 0.345, 0.345}

\usepackage{framed}
\makeatletter
 {\par\unskip\endMakeFramed%
 \at@end@of@kframe}
\makeatother

\definecolor{shadecolor}{rgb}{.97, .97, .97}
\definecolor{messagecolor}{rgb}{0, 0, 0}
\definecolor{warningcolor}{rgb}{1, 0, 1}
\definecolor{errorcolor}{rgb}{1, 0, 0}

\usepackage{alltt}
\usepackage{geometry,latexsym,graphics,authblk,color,amsmath,amsfonts,amssymb,amsbsy,amsthm,natbib,enumerate,url,multirow,threeparttable}
\usepackage[colorlinks,
            linkcolor=blue,
            anchorcolor=blue,
            citecolor=blue]{hyperref}
\usepackage{float}
\usepackage{algorithm,algpseudocode}

\usepackage{tcolorbox}
\usepackage{setspace}
\newtheorem{thm}{Theorem}
\newtheorem{lem}{Lemma}

\newcommand\R{{\sf I\kern-0.1em R}}

\newcommand{\mathsym}[1]{{}}

\newcommand{\revD}[1]{{\color{red} #1}}

\IfFileExists{upquote.sty}{\usepackage{upquote}}{}
\begin{document}
\title{An improved approximation algorithm for maximizing a DR-submodular function over a convex set}
\author[1]{Donglei Du}
\affil[1]{Faculty of Management, University of New Brunswick, Fredericton, New Brunswick,  E3B 9Y2, Canada\\ E-mail: {\tt ddu@unb.ca} }
\author[2]{Zhicheng Liu}
\affil[2]{College of Taizhou, Nanjing Normal University, Taizhou, 225300, P. R. China\\  E-mail: {\tt manlzhic@163.com} }
\author[3]{Chenchen Wu~\thanks{Corresponding author}}
\affil[3]{College of Science, Tianjin University of Technology, No.391, Binshui West Street,
Xiqing District, Tianjin, P. R. China\\ E-mail: {\tt wu\_chenchen\_tjut@163.com} }
\author[4]{Dachuan Xu}
\affil[4]{Beijing Institute for Scientific and Engineering Computing, Beijing University of Technology, Beijing 100124, P.R. China\\ E-mail: {\tt  xudc@bjut.edu.cn} }
\author[5]{Yang Zhou}
\affil[5]{School of Mathematics and Statistics, Shandong Normal University, Jinan, Shandong, P. R.China\\ E-mail: {\tt zhyg1212@163.com}}


\maketitle

\begin{abstract} Maximizing a DR-submodular function subject to a general convex set is an NP-hard problem arising from many applications in combinatorial optimization and machine learning. While it is highly desirable to design efficient approximation algorithms under this general setting where neither the objective function is monotonic nor the feasible set is down-closed, unfortunately, \cite{Vondrak2013} shows that such a problem admits no constant approximation under the value oracle model. Our main contribution is to present a $\frac{1}{4}=0.25$-approximation Frank-Wolfe type of algorithm with a sub-exponential time-complexity under the value oracle model, improving a previous approximation ratio of $\frac{1}{3\sqrt{3}}\approx 0.1924$ with the same order of time-complexity by \cite{durr2021non}.
\end{abstract}

\section{Introduction}
The main problem we are interested in this work is as follows:
\[
x^*\in\arg\max_{x\in P\subseteq [0,1]^n} F(x).
\]
In the above problem, we make the following assumptions:
\begin{enumerate}[(i)]
\item $P\subseteq [0,1]^n$ is a convex set, $\mathbf{0}=(0,\ldots, 0)\in P$ and $F(\mathbf{0})=0$.
\item $F:P\mapsto \mathbb{R}^{+}$ is $L$-smooth: $\forall x, y\in P$,
\[
|F(y)- F(x)|-\langle\nabla F(x), y-x\rangle\le \frac{L}{2}\|y-x\|^{2}\iff \|\nabla F({x})-\nabla F({y})\| \leq L\|{x}-{y}\|.
\]
\[
\bigg\Updownarrow
\]
\begin{eqnarray*}
F(x)+\langle\nabla F(x), y-x\rangle- \frac{L}{2}\|y-x\|^{2} \le F(y)\le F(x)+\langle\nabla F(x), y-x\rangle+ \frac{L}{2}\|y-x\|^{2}
\end{eqnarray*}
\item $F:[0,1]^n\mapsto \mathbb{R}^{+}$ is a non-negative DR-submodular function: $\forall x\ge y\in [0,1]^n; \forall a>0: x+ae_i, y+ae_i\in [0,1]^n$,
\[
F(x+ae_i)-F(x)\ge F(y+ae_i)-F(y), \forall i=1,\ldots, n
\]
\[
\bigg\Updownarrow
\]
\[
\nabla F(x)\ge \nabla F(y).
\]
\item Three oracles are available: (i) computing the function value $F(x)$ at any given point $x\in P$, (ii) computing its gradient $\nabla F(x)$ at any given point $x\in P$, and (iii) maximizing a linear function over $P$.\qed
\end{enumerate}

Note that previous research studying similar problems have been assuming either monotonicity of the objective function $F$ and/or down-closeness of the feasible set $P$. The main reasons are two-fold: on the negative side, lack of either assumption usually means no constant approximation algorithm exists under the value oracle model, as shown by \cite{Vondrak2013}; and on the positive side, constant approximation algorithms normally exist when one of the two assumptions holds for submodular maximization problems under many different constraints (see e.g., \cite{buchbinder2018submodular} for an excellent survey on this line of research among many others.).

In this work, we consider the problem without any of the two assumptions. This problem was first studied in \citep{durr2021non} who design an $\alpha$-approximation algorithm with convergence rate $\beta$:
\[
F(x)\ge \alpha F(x^*)-\beta.
\]
In particular, \citet{durr2021non} show that, under the previous assumptions (i)-(iv), there exists a Frank-Wolfe type of algorithm with
\[
\alpha=\frac{1}{3\sqrt{3}}\approx 0.1924, \beta=O\left(\frac{nL}{\ln^2 T} \right).
\]
where $T$ is the number iterative steps in the algorithm. Therefore, after $T=O\left(e^{\sqrt{\frac{nL}{\varepsilon}}}\right)$ iterations, we can obtain $\beta=\varepsilon$. Note that the sub-exponential convergence rate (namely, $T=2^{o(n)}$) is necessary due to the aforementioned inapproximability result in \citep{Vondrak2013}.

This work contributes along this line of research to improve the approximation ratio to 1/4 while maintaining the same order of convergence rate:
\begin{equation}\label{eq:intro_main}
\alpha=\frac{1}{4}=0.25, \beta=O\left(\frac{nL}{\ln^2 T} \right).
\end{equation}
Therefore, after $T=O\left(e^{\sqrt{\frac{nL}{\varepsilon}}}\right)$ iterations, our algorithm makes the $\beta$-term arbitrarily small: $\beta=\varepsilon$ with an improved approximation ratio of 1/4.

This problem finds applications in price optimization \citep{ito2016large} and revenue maximization on a (undirected) social network graph (a.k.a., optimal seeding) \citep{soma2017non}. Interested readers are refereed to \citep{niazadeh2018optimal,bian2020continuous} and references therein for many applications of this type of submodular maximization problems.

The rest of this paper is organized as follows.
In Section~\ref{sec:fw_dis}, we present the Frank-Wolfe algorithm and analyze its approximation ratio along with its convergence rate. Concluding remarks are provided in Section~\ref{sec:conclusion}

\section{The Frank-Wolfe algorithm}\label{sec:fw_dis}

Given a DR-submodular function $F:[0,1]^n\mapsto \mathbb{R}$, and a convex set $P\subseteq [0,1]^n$, the Frank-Wolfe algorithm is summarized as follows:
\begin{algorithm}[H]
    \caption{Frank-Wolfe algorithm}
    \label{alg:FW}
    \begin{algorithmic}[1] 
        \Procedure{FW-Dis}{$F,P, T$} 
        \State $x(0)=\mathbf{0}=(0,\ldots, 0)\in\mathbb{R}^n$ \Comment{Initialization: $\mathbf{0}\in P$ due to assumption}
        \For{$j=0,1\ldots, T-1$} \Comment{$T$ is the number of iterations}
        \begin{eqnarray}
        \label{eq:fw_discrete3}t_j&=&\frac{1}{1-\ln \left(1+\frac{H_j}{H_T}\right)}-1 \\
        \label{eq:fw_discrete2} v(x(t_j))&\in&\arg\max_{v\in P}\left\langle \nabla F(x (t_j)),v\right\rangle\\
        \label{eq:fw_discrete1}x(t_{j+1})&=&e^{-\frac{1}{1+t_j}+\frac{1}{1+t_{j+1}}}x(t_j)+\left(1-e^{-\frac{1}{1+t_j}+\frac{1}{1+t_{j+1}}}\right)v(x(t_j))
      \end{eqnarray}
\EndFor\label{FW-Disfor}
            \State \textbf{return} $ x \in \arg\max_{x(t_j)} F(x (t_j)) $ \Comment{$t_T=\frac{\ln 2}{1-\ln 2}\approx 2.258891$}
        \EndProcedure
    \end{algorithmic}
\end{algorithm}

The following comments are in order:
 \begin{enumerate}[(1)]
\item To save space in the future analysis, denote
 \begin{eqnarray*}
\sqrt{a_{t_j}}&:=&e^{\frac{t_j}{1+t_j}}.
\end{eqnarray*}
 \item Note that
 \begin{eqnarray}
\label{eq:fw_discrete_sqrt_a_t} (\ref{eq:fw_discrete3})&\iff &\sqrt{a_{t_j}}:=e^{\frac{t_j}{1+t_j}}=1+\frac{H_j}{H_T},
\end{eqnarray}
where $H_j$ is the harmonic number:
\begin{eqnarray}
\label{eq:fw_discrete_harmonic1}H_{T}&=&\sum_{j=1}^T\frac{1}{j}=\Theta(\ln T),\\
\label{eq:fw_discrete_harmonic2} H_{2,T}&=&\sum_{j=1}^T\frac{1}{j^2}=\Theta(1).
\end{eqnarray}
We will also use the last quantity $H_{2,T}$, which is the finite sum of the reciprocals of the first $T$ squared natural numbers, in the convergence rate analysis (see the proof of Theorem~\ref{thm:main_result_discrete} in Section~\ref{subsec:approx_discrette}.).
\item Moreover, the step-size in (\ref{eq:fw_discrete1}) and the notation $a_t$ are related as follows:
\[
e^{-\frac{1}{1+t_j}+\frac{1}{1+t_{j+1}}}=\sqrt{\frac{a_{t_j}}{a_{t_{j+1}}}}\in [0,1]
\]

 \item $T$ is the total number of iterations.
 \item $t_j$ in (\ref{eq:fw_discrete3}) is an increasing function of $j$ and indicates the time value of $t$ at the start of the $j$-th iteration, satisfying the boundary conditions
 \[
 t_0=0, t_T=\frac{1}{1-\ln 2}-1=\frac{\ln 2}{1-\ln 2}\approx 2.258891.
 \]
  \item The solution at each step is feasible as it is the convex combination of feasible solutions $x(t_j), v(t_j)\in P$.
 \end{enumerate}

\subsection{Analysis}\label{sec:dis_analysis}

\subsubsection{Two lemmas}\label{subsec:lemmas_discrete}
We first bound $x_i(t_j)$ from above.
\begin{lem}\label{lem:dis_x_coor} The algorithm maintains the following inequality:
\[
1-x_i(t_{j})\ge e^{-\frac{t_{j}}{1+t_{j}}}=\frac{1}{\sqrt{a_{t_j}}}, i=1,\ldots, n; j=0, 1,\ldots, T.
\]
\end{lem}
\begin{proof}
From (\ref{eq:fw_discrete1}) in the algorithm, we have
\begin{eqnarray}
\nonumber 1-x_i(t_{j+1})&=&1-e^{-\frac{1}{1+t_j}+\frac{1}{1+t_{j+1}}}x_i(t_j)-\left(1-e^{-\frac{1}{1+t_j}+\frac{1}{1+t_{j+1}}}\right)v(x(t_j))\\
\nonumber &\overset{v(x(t_j))\in P\Longrightarrow v_i(x(t_j))\le 1}{\ge}&
1-e^{-\frac{1}{1+t_j}+\frac{1}{1+t_{j+1}}} x(t_j)-1+e^{-\frac{t_{j+1}}{1+t_{j+1}}+\frac{t_j}{1+t_j}} \\
\nonumber &=&-e^{-\frac{1}{1+t_j}+\frac{1}{1+t_{j+1}}} x(t_j)+e^{\frac{1}{1+t_{j+1}}-\frac{1}{1+t_j}} \\
\nonumber &=&e^{-\frac{1}{1+t_j}+\frac{1}{1+t_{j+1}}}(1-x(t_j))\\
\nonumber &\ge &e^{\sum_{\ell=0}^j\left[-\frac{1}{1+t_j}+\frac{1}{1+t_{j+1}}\right]}(1-x(t_0))\\
\label{eq:sub_max_non_mono_x_coordinate_discrete}&=&e^{-\frac{t_{j+1}}{1+t_{j+1}}}=e^{-\frac{t_{j+1}}{1+t_{j+1}}}=\frac{1}{\sqrt{a_{t_{j+1}}}}
\end{eqnarray}
\end{proof}

We recall the following known result~\cite[Lemma 3.5]{feldman2011unified}, which is a consequence of the DR-submodularity of $F$. For any DR-submodular function $F: P\mapsto\mathbb{R}^{+}$, and $\forall x, x^*\in P$, where $P\subseteq [0,1]^n$, we have
\begin{equation}
\label{eq:sub_max_non_mono_theta}F(x\vee x^*)\ge\left(1-\|x\|_{\infty}\right)F(x^*)=\left(1-\max_{i=1}^nx_i(t)\right)F(x^*)
 \end{equation}

 With (\ref{eq:sub_max_non_mono_x_coordinate_discrete})-(\ref{eq:sub_max_non_mono_theta}), we are able to bound $F(x(t)\vee x^*)$ from below, which is summarized in the following lemma.
 \begin{lem}\label{lemma:sub_max_non_mono_f_x_vee_x_star_dis}
\begin{equation}\label{eq:fx_meet_opt_lb_discrete}
F(x(t_j)\vee x^*)\ge e^{-\frac{t_j}{1+t_j}}F(x^*)=\frac{1}{\sqrt{a_{t_j}}}F(x^*)
\end{equation}
\end{lem}
\begin{proof}
\begin{eqnarray*}
F(x(t_j)\vee x^*)&\overset{(\ref{eq:sub_max_non_mono_theta})}{\ge} &\left(1-\max_{i=1}^nx_i(t_j)\right)F(x^*)\\
&\overset{(\ref{eq:sub_max_non_mono_x_coordinate_discrete})}{\ge}&e^{-\frac{t_j}{1+t_j}}
F(x^*)=\frac{1}{\sqrt{a_{t_j}}}F(x^*)
 \end{eqnarray*}
 \end{proof}

 \subsubsection{Bounding the increment of the Lyapunov function from below}\label{subsec:E_t_discrete}

We introduce the following Lyapunov function analyze the approximation ratio and the convergence rate of Algorithm \textsc{FW-Dis}:
\begin{eqnarray}
\label{eq:fw_discrete_Lyapunov_function} E(t_j) &=&a_{t_j}F(x (t_j))-\sqrt{a_{t_j}}F(x^*), j=0,1\ldots, T,
\end{eqnarray}

We show that the increment of Lyapunov-like function in (\ref{eq:fw_discrete_Lyapunov_function}) at each iteration is bounded from below.

Below we will use the following known bound which is implied by any DR-submodular function~\citep{hassani2017gradient}.  For any DR-submodular function $F: P\mapsto\mathbb{R}^{+}$, where $P\subseteq [0,1]^n$, we have
\begin{equation}\label{eq:dr_lb}
\langle\nabla F(x), y-x\rangle \ge F(x \vee y)+F(x \wedge y)-2 F(x), \forall x, y\in P
\end{equation}

\begin{proof}
We can always assume that $E(t_j) \le 0$ is true for $\forall j \in [T-1]$, otherwise we can notice that $F(x(t_j)) \ge \frac{1}{\sqrt{a_{t_j}}} F(x^\ast) > \frac{1}{4}F(x^\ast)$ for some $j \in [T-1]$ so that the approximation guarantee is already achieved.
\begin{eqnarray}
\nonumber&&E(t_{j+1})-E(t_j)\\
\nonumber&=&a_{t_{j+1}}F(x(t_{j+1}))-a_{t_j}F(x(t_j))-\left[\sqrt{a_{t_{j+1}}}-\sqrt{a_{t_j}}\right]F(x^*)\\
\nonumber&=&a_{t_{j+1}}\left[F(x(t_{j+1}))-F(x(t_j))\right]+[a_{t_{j+1}}-a_{t_j}]F(x(t_j))-\left[\sqrt{a_{t_{j+1}}}-\sqrt{a_{t_j}}\right]F(x^*)\\
\nonumber&\overset{L-\text{smooth}}{\ge}&a_{t_{j+1}}\left[\left(1-\sqrt{\frac{a_{t_j}}{a_{t_{j+1}}}}\right)\langle \nabla F(x(t_j)), v(x(t_j)-x(t_j))\rangle-\frac{L}{2}\left(1-\sqrt{\frac{a_{t_j}}{a_{t_{j+1}}}}\right)^2\|v(x(t_j)-x(t_j))\|^2 \right]\\
\nonumber&&+[a_{t_{j+1}}-a_{t_j}]F(x(t_j))-\left[\sqrt{a_{t_{j+1}}}-\sqrt{a_{t_j}}\right]F(x^*)\\
\nonumber&\overset{\|v(x(t_j))-x(t_j)\|^2\le n}{\ge}&a_{t_{j+1}}\left[\left(1-\sqrt{\frac{a_{t_j}}{a_{t_{j+1}}}}\right)\langle \nabla F(x(t_j)), v(x(t_j)-x(t_j))\rangle-\frac{nL}{2}\left(1-\sqrt{\frac{a_{t_j}}{a_{t_{j+1}}}}\right)^2\right]\\
\nonumber&&+[a_{t_{j+1}}-a_{t_j}]F(x(t_j))-\left[\sqrt{a_{t_{j+1}}}-\sqrt{a_{t_j}}\right]F(x^*)
\end{eqnarray}

\begin{eqnarray}
\nonumber&\overset{(\ref{eq:fw_discrete2})}{\ge}&a_{t_{j+1}}\left[\left(1-\sqrt{\frac{a_{t_j}}{a_{t_{j+1}}}}\right)\langle \nabla F(x(t_j)), x^*-x(t_j))\rangle-\frac{nL}{2}\left(1-\sqrt{\frac{a_{t_j}}{a_{t_{j+1}}}}\right)^2\right]\\
\nonumber&&+[a_{t_{j+1}}-a_{t_j}]F(x(t_j))-\left[\sqrt{a_{t_{j+1}}}-\sqrt{a_{t_j}}\right]F(x^*)\\
\nonumber&\overset{(\ref{eq:dr_lb})}{\ge}&a_{t_{j+1}}\left[\left(1-\sqrt{\frac{a_{t_j}}{a_{t_{j+1}}}}\right)\left(F(x^*\vee x(t_j))- 2F(x(t_j))\right)-\frac{nL}{2}\left(1-\sqrt{\frac{a_{t_j}}{a_{t_{j+1}}}}\right)^2\right]\\
\nonumber&\overset{(\ref{eq:fx_meet_opt_lb_discrete})}{\ge}&a_{t_{j+1}}\left[\left(1-\sqrt{\frac{a_{t_j}}{a_{t_{j+1}}}}\right)\left(\frac{1}{\sqrt{a_{t_j}}}F(x^*)- 2F(x(t_j))\right)-\frac{nL}{2}\left(1-\sqrt{\frac{a_{t_j}}{a_{t_{j+1}}}}\right)^2\right]\\
\nonumber&&+[a_{t_{j+1}}-a_{t_j}]F(x(t_j))-\left[\sqrt{a_{t_{j+1}}}-\sqrt{a_{t_j}}\right]F(x^*)\\
\nonumber&=&\sqrt{a_{t_j}}\left[\left(1-\sqrt{\frac{a_{t_j}}{a_{t_{j+1}}}}\right)\frac{a_{t_{j+1}}}{a_{t_j}}-\sqrt{\frac{a_{t_{j+1}}}{a_{t_j}}}+1\right]F(x^*) - a_{t_j}\left[-2\sqrt{\frac{a_{t_{j+1}}}{a_{t_j}}}+\frac{a_{t_{j+1}}}{a_{t_j}}+1\right]F(x(t_j))\\
\nonumber&&-a_{t_{j+1}}\frac{nL}{2}\left(1-\sqrt{\frac{a_{t_j}}{a_{t_{j+1}}}}\right)^2\\
\nonumber&=&\sqrt{a_{t_j}}\left[\frac{a_{t_{j+1}}}{a_{t_j}}-2\sqrt{\frac{a_{t_{j+1}}}{a_{t_j}}}+1\right]F(x^*) - a_{t_j}\left[\frac{a_{t_{j+1}}}{a_{t_j}}-2\sqrt{\frac{a_{t_{j+1}}}{a_{t_j}}}+1\right]F(x(t_j))\\
\nonumber&&-\frac{nL}{2}\left(\sqrt{a_{t_{j+1}}}-\sqrt{a_{t_j}}\right)^2\\
\nonumber&=&\sqrt{a_{t_j}}\left(\sqrt{\frac{a_{t_{j+1}}}{a_{t_j}}}-1\right)^2 \left(F(x^*) -\sqrt{a_{t_j}}F(x(t_j))\right)-\frac{nL}{2}\left(\sqrt{a_{t_{j+1}}}-\sqrt{a_{t_j}}\right)^2\\
\label{eq:e_t_dis_lb}
&\overset{E(x(t_j)) \le 0}\ge& -\frac{nL}{2}\left(\sqrt{a_{t_{j+1}}}-\sqrt{a_{t_j}}\right)^2
\end{eqnarray}

\end{proof}

\subsubsection{Approximation ratio and cconvergence rate}\label{subsec:approx_discrette}
\begin{thm}\label{thm:main_result_discrete} Algorithm \textsc{FW-Dis} outputs a solution $x_{t_T}$ satisfying
\[
F(x(t_T))\ge \frac{1}{4}F(x^*)-O\left(\frac{nL}{\ln^2 T} \right)
\]
after $T$ iterations. Consequently,
\[
F(x(t_T))\ge \frac{1}{4}F(x^*)-\varepsilon
\]
after $T=O\left(e^{\sqrt{\frac{nL}{\varepsilon}}}\right)$ iterations.
\end{thm}
\begin{proof} When Algorithm \textsc{FW-Dis} terminates, we have
\[
t_T=\frac{1}{1-\ln 2}-1\approx 2.258891,
\]
and hence
\[
a_0=1, \sqrt{a_{t_T}}=e^{\frac{t_T}{1+t_T}}=e^{\frac{\left(\frac{1}{1-\ln 2}-1\right)}{\frac{1}{1-\ln 2}}}=2.
\]
Therefore summing up inequality (\ref{eq:e_t_dis_lb}) over $j=0,1,\ldots, T$, we obtain
\begin{eqnarray*}
E(t_T)-E(t_0)&\overset{F(0)=0}{=}&a_{t_T}F(x(t_T))-(\sqrt{a_{t_T}}-1)F(x^*)\\
&=&4F(x(t_T))-F(x^*)\\
&\overset{(\ref{eq:e_t_dis_lb})}{\ge}& -\frac{nL}{2}\sum_{j=0}^{T-1}\left(\sqrt{a_{t_{j+1}}}-\sqrt{a_{t_j}}\right)^2\\
&\overset{(\ref{eq:fw_discrete_sqrt_a_t})}{=}&-\frac{nL}{2}\sum_{j=0}^{T-1}\left(1+\frac{H_{j+1}}{H_T}-1-\frac{H_{j}}{H_T}\right)^2=-\frac{nL}{2}\frac{1}{H^2_T}\sum_{j=1}^{T}\frac{1}{j^2}=-\frac{nL}{2}\frac{H_{2,T}}{H^2_{T}}\\
&\overset{(\ref{eq:fw_discrete_harmonic1})-(\ref{eq:fw_discrete_harmonic2})}{=}&-O\left(\frac{nL}{\ln^2 T} \right).
\end{eqnarray*}
So we have
\[
F(x(t_T))\ge \frac{1}{4}F(x^*)-O\left(\frac{nL}{\ln^2 T} \right).
\]
Consequently,
\[
F(x(t_T))\ge \frac{1}{4}F(x^*)-\varepsilon,
\]
after $T=O\left(e^{\sqrt{\frac{nL}{\varepsilon}}}\right)$ iterations..
\end{proof}

As another corollary, if $F$ is the multilinear extension of a submodular function, then \citet[Lemma 2.3.7.]{feldman2013maximization} shows that $F$ is $L$-smooth with
\[
L=O\left(n^2\right)F(x^*)
\]
Plugging this $L$ into the above, we have
\[
F(x(t_T))\ge \left(\frac{1}{4}-O\left(\frac{n^3}{\ln^2 T}\right)\right)F(x^*).
\]
Consequently,
\[
F(x(t_T))\ge \left(\frac{1}{4}-\varepsilon\right)F(x^*)
\]
after $T=O\left(e^{\sqrt{\frac{n^3}{\varepsilon}}}\right)$ iterations.

\section{Conclusions}\label{sec:conclusion}

The assumptions of $0\in P$ and $F(0)=0$ can both be relaxed. Instead of initializing both algorithms with the starting point $x(0)=\mathbf{0}$, we use $x(0)=\arg\min_{x\in P}\|x\|_{\infty}$. Then almost the same arguments can be carried to yield the following starting-point-dependent approximation ratio. We have
\[
F(x(t_T))-\frac{1}{4}F(0)\ge \frac{1}{4}\left(1-\min_{x\in P}\|x\|_{\infty}\right)F(x^*)-O\left(\frac{nL}{\ln^2 T} \right).
\]

%
%

An open question is whether the approximation $1/4$ in (\ref{eq:intro_main}) is the best possible with a sub-exponential convergent rate?


 \section*{Acknowledgements} The first author's research is supported by the NSERC grant (No. 283106), and NSFC grants (Nos. 11771386 and  11728104). The third author's research is supported by the NSFC grants (No. 11971349). The fourth author's research is supported by the NSFC grants (No. 12131003) and Beijing Natural Science Foundation Project (No. Z200002).   The fifth author's research is supported by the NSFC grants (No. 12001335).

\bibliographystyle{apalike}

\bibliography{mybibfile}

\end{document}